\documentclass[letterpaper, 10 pt, conference]{ieeeconf}

\usepackage{cite}
\usepackage{amsmath,amssymb,amsfonts}
\usepackage{algorithmic}
\usepackage{graphicx}
\usepackage{textcomp}
\usepackage{optidef}
\usepackage{jphmacros2e} 
\usepackage{physics}
\usepackage{breqn}
\usepackage{array}
\usepackage{tabularx}
\usepackage{collcell}
\usepackage{booktabs}
\usepackage{balance}
\usepackage{subcaption}

\usepackage[unicode=true,bookmarks=false,breaklinks=false,pdfborder={0 0 1},colorlinks=false]{hyperref}
\hypersetup{colorlinks,bookmarksopen,bookmarksnumbered,citecolor=red,urlcolor=blue}

\renewcommand\qedsymbol{$\blacksquare$}
\usepackage{xcolor}
\def\BibTeX{{\rm B\kern-.05em{\sc i\kern-.025em b}\kern-.08em
    T\kern-.1667em\lower.7ex\hbox{E}\kern-.125emX}}
\begin{document}
\IEEEoverridecommandlockouts

\title{Quantum Deception: Honey-X Deception using Quantum Games}

\author{Efstratios~Reppas$^1$,
Ali~Wadi$^2$,~\IEEEmembership{Student Member,~IEEE,}
Brendan Gould$^1$,~\IEEEmembership{Student Member,~IEEE,}\\
and~Kyriakos~G.~Vamvoudakis$^2$,~\IEEEmembership{Senior Member,~IEEE}
\thanks{$^1$E. Reppas and B. T. Gould are with the School of Electrical and Computer Engineering, Georgia Institute of Technology, Atlanta, GA, USA. Email: {\tt\small ereppas3@gatech.edu; bgould6@gatech.edu}.}
\thanks{$^2$A. Wadi and K. G. Vamvoudakis are with the Daniel Guggenheim School of Aerospace Engineering, 
Georgia Institute of Technology, Atlanta, GA, 
USA. Email:
{\tt\small awadi@gatech.edu; kyriakos@gatech.edu}.}
\thanks{This work was supported in part by NSF under grant Nos.  CPS-$2227185$ and SATC-$2231651$, and by ARO under grant No. W$911$NF-$24-1-0174$.}
}

\markboth{ }%
{ }

\maketitle

\begin{abstract}
In this paper, we develop a framework for deception in quantum games, extending the Honey-X paradigm from classical zero-sum settings into the quantum domain. Building on a view of deception in classical games as manipulation of a player’s perception of the payoff matrix, we formalize quantum deception as controlled perturbations of the payoff Hamiltonian subject to a deception budget. We show that when victims are aware of possible deception, their equilibrium strategies surprisingly coincide with those of naive victims who fully trust the deceptive Hamiltonian. This equivalence allows us to cast quantum deception as a bilevel optimization problem, which can be reformulated into a bilinear semidefinite program. To illustrate the framework, we present simulations on quantum versions of the Penny Flip game, demonstrating how quantum strategy spaces and non-classical payoffs can amplify the impact of deception relative to classical formulations. 
\end{abstract}



\section{Introduction}

Quantum game theory is a generalization of classical game theory that extends the analysis of strategic interactions into the quantum domain. 
In classical game theory, player strategies are evaluated using payoff functions that map combinations of individual strategies to real numbers~\cite{hespanha2017noncooperative}. 
Quantum game theory, on the other hand, models strategies as quantum operators that act on the state of a quantum system. 
This expansion of the strategy space---often significantly larger than in classical settings---introduces new possibilities and strategic advantages~\cite{CHEON2006147, Accardi20212, Boukas2000, Accardi2022, PhysRevA.84.012311}.

Quantum operators are distinguished by their mathematical representations. 
While classical games typically use tensors to associate strategies with payoffs, quantum games involve non-linear mappings and are often described using system-level payoff matrices with non-zero off-diagonal elements~\cite{Carlen2024, quiroga2025quantumeigengameexcitedstate, abbas2021power}. 
This feature is absent in classical formulations. 
The quantum framework not only captures scenarios where players manipulate quantum systems but also provides a unified perspective that includes classical games as a special case.

Within classical games, the structure of \emph{information} often influences both player actions and game outcomes~\cite{sobel2020signaling,gouldInformationDesignUncertainty2023, kamenicaBayesianPersuasion2011}. 
Even when restricted to truthful information, an intelligent designer can incentivize individual agents to take specific actions~\cite{bergemannInformationDesignUnified2019}, or affect outcomes for the entire population~\cite{acemogluInformationalBraessParadox2018}, simply by strategically choosing what to reveal and what to conceal. 
Given these results, it is natural to ask what becomes possible when deceptive information may also be used. 
Several types of deception have been identified~\cite{pawlickGametheoreticTaxonomySurvey2019}, with potential applications in both physical- and cyber-security scenarios~\cite{tambeSecurityGameTheory2011,nguyenImitativeAttackerDeception2019}. 

In this paper we will introduce a concept of deception in quantum games. To the best of our knowledge, our results provide the first foundation for studying adversarial information design in quantum strategic interactions. Specifically, we will investigate how the proposed deception framework from~\cite{gouldNovelFrameworkHoneyX2025} can be applied to quantum games.
Beyond the theoretical interest of introducing deception in quantum games, there are concrete cybersecurity scenarios that motivate this line of work. Modern honeypots and deception systems aim to increase attacker uncertainty and gather forensic information by presenting believable decoy services; recent work shows that generative and adaptive AI can dramatically improve honeypot realism and effectiveness \cite{pu2022honeypot}. While quantum computers pose a distinct threat to cryptographic primitives, deception is complementary: it increases an attacker’s operational cost and measurement uncertainty, and can be used to mask real assets, slow reconnaissance, and funnel attackers into monitored environments. In particular, adversarial generative inference techniques (e.g., ALI/BiGAN variants \cite{dumoulin2016ali,donahue2016bigan}) are natural tools for building adaptive decoys and realistic synthetic interactions that could be deployed inside classical or hybrid quantum-classical honeypots. We therefore view our theoretical results as underpinning a program of applied work that couples deception, post-quantum cryptographic hardening \cite{NISTPQC2022,chen2021classic}, and adaptive generative modeling to increase practical resilience against advanced (including quantum-capable) adversaries.

\subsection*{Related Work}
Early quantum game theory models (e.g., Meyer~\cite{Meyer_1999}, EWL~\cite{EWL}) are criticized for imposing artificial constraints on strategy spaces, which lack physical justification and undermine claims of quantum advantage. Zhang~\cite{Zhang2012} shows some constraints, like the inverse joint operation, are unnecessary. The field also faces broader issues like the lack of general frameworks, overemphasis on qualitative analysis, and unresolved conceptual ambiguities in quantizing classical games~\cite{FLITNEY2007381, Lee2003, UnderstandingGameTheoryV2}. These challenges are addressed in the framework proposed in~\cite{wu2005hamiltotianformalismgametheory}, which aims to resolve some of the criticisms by formulating a consistent approach to treating quantum games.

Regarding deception, in prior work it is common to use \emph{signaling games}~\cite{pawlickModelingAnalysisLeaky2019,hespanhaSensorManipulationGames2021} to study how deceptive design of messages can manipulate rational agents, even when those agents know they are being lied to. 
This type of deception, which presents false information to victims with the goal of influencing them to take actions beneficial to the deceiver, is called Honey-X due to its similarity to ``honeypot'' servers in cybersecurity~\cite{pibilGameTheoreticModel2012}. 
Recent work~\cite{gouldNovelFrameworkHoneyX2025} proposed a generalization of Honey-X deception in zero-sum, classical games that both allowed for additional types of deception and explicitly considered possible meta-rational reasoning by the victim.
To the best of our knowledge, no deception methods exist that are explicitly designed to apply to quantum games. 

\textit{Contributions:}  
The contribution of the present paper is multifold. 
First, we generalize the deceptive framework of~\cite{gouldNovelFrameworkHoneyX2025} from classical zero-sum games to the quantum domain. 
Second, we show that despite the enlarged strategy space induced by quantum moves, several key properties of the classical model remain intact. 
In particular, we prove that the security policy of a robust victim---one who accounts for deception within a bounded budget---coincides with the security policy of a naive victim who fully trusts the presented game. 
Finally, leveraging this equivalence, we formulate a mathematical program for computing optimal quantum deception and the corresponding security strategies for both the deceiver and the victim, in direct analogy to the classical construction.


\paragraph*{Structure} The remainder of this paper is structured as follows. Section~\ref{sec:problem_formulation} formulates quantum games and recalls a model of deception in classical games that is subsequently extended to the quantum formulation. Then, Section~\ref{sec:resolving} develops a convex optimization which computes the optimal deception and victim response strategies, and Section~\ref{sec:simulations} presents numerical simulations that illustrate the impact and advantage of our proposed quantum deception. Finally, Section~\ref{sec:conclusion} concludes the paper and discusses future directions.  

\paragraph*{Notation}  
The set $\mathbb{R}$ denotes the set of real numbers, while $\mathbb{N}$ denotes the set of natural numbers including zero. 
For a matrix $A$, $A^{\otimes n}$ denotes the $n^{\text{th}}$ Kronecker power of $A$. 
For a symmetric matrix $A$, $\lambda_{\min}(A)$ denotes the minimum eigenvalue of $A$. 
The matrix $I_n$ denotes the identity matrix of order $n$.  
Given a (complex) Hilbert space $\mathcal{H}$, we will use Dirac's bra-ket notation and write $\ket{\psi}\in\mathcal{H}$ for an element and $\bra{\psi}=\ket{\psi}^\dag$ for its adjoint.
We will also write $\mathbb{H}^d$ for the space of $d\times d$ Hermitian matrices, and $\mathbb{H}_+^d$ for the cone of positive semidefinite matrices in $\mathbb{H}^d$.  
We denote by $\norm{\mathbf{A}} = \sqrt{\tr\!\left(\mathbf{A}^{\dag}\mathbf{A}\right)}$ the Frobenius norm of $\mathbf{A}$ in $\mathbb{H}^d$.  
$\mathcal{U}(n)\subset \mathbb{C}^{n\times n}$ denotes the set of all $n\times n$ unitary matrices.  
$\mathcal{H}$ is a finite-dimensional complex Hilbert space equipped with the inner product $\braket{A}{B} = \tr(A^\dagger B)$. 

\section{Problem Formulation}\label{sec:problem_formulation}

First, Section~\ref{subsec:quantum_game} introduces a framework that generalizes classical matrix games to incorporate quantum strategies and quantum payoffs. 
Then, in Section~\ref{subsec:deception_model}, we establish the connection between classical and quantum settings in order to define our model of \emph{Honey-X deception} within the quantum framework.

\subsection{Quantum Games}\label{subsec:quantum_game}

A \emph{quantum game} is specified by the quadruple
\begin{equation*}
\qty{\mathcal{N}, \rho_0, \left(S_i\right)_{i\in\mathcal{N}}, \left(H_i\right)_{i\in\mathcal{N}}},
\end{equation*}
where, $\mathcal{N} = \{1, \dots, N\}$ is the set of players,
     $\mathcal{H}_i$ is the finite-dimensional Hilbert space of player $i$, $\rho_0 \in \bigotimes_i \mathcal{H}_i$ is the initial joint state,
     $S_i$ is the set of admissible quantum strategies for player $i$, with each $U_j \in S_i$ a unitary acting on $\mathcal{H}_i$, and
  $H_i$ is the Hermitian payoff operator for player $i$.
The expected payoff for player $i$ is given by
\begin{equation*}
    u_i(U_1, \dots, U_N) = 
    \tr\!\left( H_i \left(\bigotimes_{i=1}^N U_i\right)\rho_0\left(\bigotimes_{i=1}^N U_i\right)^\dagger \right).
\end{equation*}
A profile of strategies $(U_1^\star, \dots, U_N^\star)$ is a Nash equilibrium if, for every player $i$ and every $U_i \in S_i$,
\begin{equation}
    \label{eq:quantum_ne}
    u_i(U_i; U_{-i}^\star) \leq u_i(U_i^\star). 
\end{equation}


We restrict attention to the case of \emph{unentangled quantum states}. 
That is, the final state $\rho_f$ is separable and can be expressed as
\begin{equation*}
\begin{split}
    \rho_f 
    &= \left(\bigotimes_{i=1}^N U_i\right)\rho_0\left(\bigotimes_{i=1}^N U_i\right)^\dagger \\
    &= \rho_1 \otimes \rho_2 \otimes \cdots \otimes \rho_N.
\end{split}
\end{equation*}

In the two-player case, a quantum game is characterized by a quantum payoff matrix $H \in \mathbb{C}^{(n_A n_B) \times (n_A n_B)}$. 
Each player's actions are described by positive semi-definite, trace-one density matrices $\rho_A \in \mathbb{C}^{n_A \times n_A}$ and $\rho_B \in \mathbb{C}^{n_B \times n_B}$. 
The resulting payoff is~\cite{wu2005hamiltotianformalismgametheory}:
\begin{align*}
    u_{\textrm{quantum}} = \operatorname{tr}\qty((\rho_A \otimes \rho_B)H).
\end{align*}
The Hermitian constraint on $H$ guarantees real-valued payoffs, while the density matrix constraints $\rho \succeq 0$ and $\tr(\rho)=1$ ensure valid quantum states, with diagonal entries corresponding to probabilities. We denote by $\mathcal{P}=\{\rho\in\mathbb{C}^{n\times n}:\tr(\rho)=1,\rho \succeq 0\}$ the set of density matrices, so $\rho_A\in\mathcal{P}_A, \rho_B\in\mathcal{P}_B$.

Equivalently, payoffs can be defined using pure states $\ket{\psi_i} \in \mathcal{H}_i$. 
For $\ket{\psi} = \sum_i c_i \ket{\psi_i}$ with $\sum_i |c_i|^2 = 1$, the expected payoff is computed as
\begin{equation*}
    u_{\textrm{quantum}} = \mel**{\psi_A \otimes \psi_B}{H}{\psi_A \otimes \psi_B}.
\end{equation*}
This formulation highlights the key advantage of quantum games: superposition and entanglement enlarge the strategic space beyond classical analogues, potentially leading to strictly improved payoffs.

To construct payoff operators, one may use the Hamiltonian formalism~\cite{Lee2003, wu2005hamiltotianformalismgametheory}. 
Specifically, let $H(i,j)$ denote the $(i,j)$-th entry of $H$, defined as
\begin{equation}\label{eq:Hqm}
\begin{split}
    H(i,j) &= \mel{\mu_1 \mu_2}{H}{\nu_1 \nu_2} \\
           &= \tr\!\left( P(\nu_1 \nu_2)\,\rho_0\,(\mu_1 \mu_2)^\dagger \right),
\end{split}
\end{equation}
where $\mu_i, \nu_i \in \mathcal{U}(n_i)$ are unitary strategies of player $i$, and $P$ is a suitable payoff operator. 
For an $n_1 \times n_2$ matrix game, we have $|S_i| = n_i^2$, $P \in \mathbb{R}^{n_1^2 \times n_2^2}$, and $H \in \mathbb{C}^{n_1^2 \times n_2^2}$.

\paragraph*{Example} 
For a classical $2 \times 2$ matrix game, the Pauli basis
\begin{equation}\label{eq:pauli}
    \begin{split}
    I &= \begin{bmatrix} 1 & 0 \\ 0 & 1 \end{bmatrix}, \quad
    X = \begin{bmatrix} 0 & 1 \\ 1 & 0 \end{bmatrix}, \\
    Y &= \begin{bmatrix} 0 & -i \\ i & 0 \end{bmatrix}, \quad
    Z = \begin{bmatrix} 1 & 0 \\ 0 & -1 \end{bmatrix},
    \end{split}
\end{equation}
spans $\mathcal{U}(2)$ and yields a quantum payoff Hamiltonian $H \in \mathbb{C}^{16 \times 16}$. \hfill\(\square\)

\begin{remark}
    If the strategy set $S_i$ corresponds to the complete basis of $\mathcal{U}(n_i)$, then the number of admissible strategies for player $i$ is $n_i^2$. \hfill\(\square\)
\end{remark}

\begin{remark}
    A classical $n$-strategy game typically lifts to a quantum game of size $O(n^2)$, reflecting the enlarged strategy space in the quantum domain. \hfill\(\square\)
\end{remark}

\begin{remark}
    The density matrices in~\eqref{eq:Hqm} represent distributions over the \emph{quantum strategy space}, not over the game payoffs themselves. \hfill\(\square\)
\end{remark}

\subsection{Honey-X Deception}\label{subsec:deception_model}

We now turn to deception. 
In the classical setting, a two-player zero-sum game is parameterized by $G \in \mathbb{R}^{m \times n}$, with mixed strategies $x \in \Delta(m)$ and $y \in \Delta(n)$, where $\Delta(\cdot)$ denotes the probability simplex of appropriate dimension, yielding payoff $x^\top G y$. 
Deception, as introduced in~\cite{gouldNovelFrameworkHoneyX2025}, allows the row player not only to select $x$ but also to perturb the column player’s perceived payoffs via $\tilde{G} = G + D$, subject to constraints on $D$, referred to as the deceptive payoff. 
The column player best-responds to $\tilde{G}$, but the actual payoffs are determined by $G$.

In the quantum setting, the analogue of this deception mechanism is to perturb the payoff Hamiltonian. 
That is, instead of $H$, the deceiver announces
\begin{equation*}
    H' = H + D,
\end{equation*}
with $D \in \mathbb{C}^{(n_A  n_B) \times (n_A  n_B)}$ Hermitian to preserve interpretability. 
As in the classical case, the deceiver’s manipulation is limited by a deception budget.

\begin{assumption}[Deception Budget]\label{asmp:budget}
    Let $\Delta \in \mathbb{R}_+$ be the deception budget. 
    We assume that $\lVert D \rVert_1 \leq \Delta$, where $\lVert \cdot \rVert_1$ denotes the matrix norm induced by the vector $1$-norm $\norm{A}_1 = \sup_{\lVert x \rVert_1 = 1} \lVert A x \rVert_1$. We denote the corresponding set that $D$ belongs to as $\mathcal{D}\coloneq \{\mathbb{C}^{(n_A  n_B) \times (n_A  n_B)}:\lVert D \rVert_1 \leq \Delta\}$.\hfill\(\square\)
\end{assumption}

\begin{assumption}[Victim Response]\label{asmp:victim_response}
    The victim observes only the perturbed Hamiltonian $H'$ and selects a Nash equilibrium strategy (as defined in~\eqref{eq:quantum_ne}) for $H'$. \hfill\(\square\)
\end{assumption}

This ``naive victim'' model, in which the victim fully trusts $H'$, might seem restrictive. 
However, as shown in the following sections, it aligns with classical deception models and naturally extends to quantum games. 
Figure~\ref{fig:timeline} illustrates the quantum deception game framework.

\begin{figure}[!t]
    \centering
    \includegraphics[width=\linewidth]{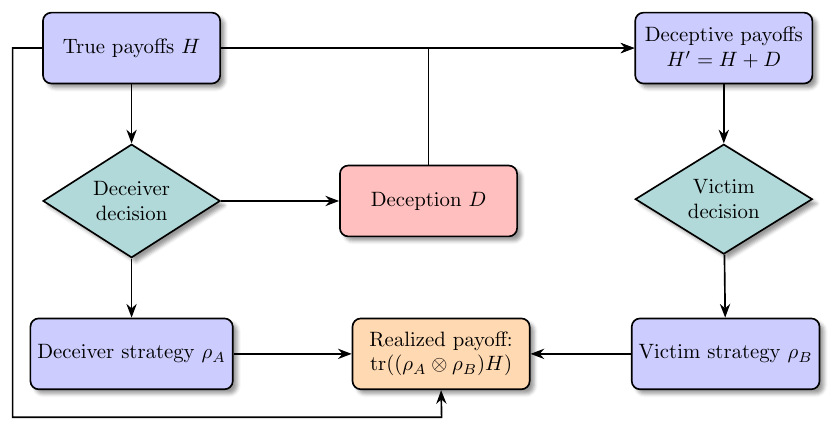}
    \caption{The schematic of the quantum deception game. Rectangles denote numerical values, while parallelograms represent strategic decisions. Arrows illustrate causal relationships, indicating that the origin influences the destination.}
    \label{fig:timeline}
\end{figure}

\section{Resolving Quantum Deception}\label{sec:resolving}

In order to begin deriving quantum deception, the first nuance we must address is how to extend the notion of best response from classical to quantum games. 

\begin{lemma}
For the quantum game with the Hermitian payoff operator $H \in \mathcal{H}_A \otimes \mathcal{H}_B$, player density operators 
$\rho_A \in \mathcal{P}_A$ and 
$\rho_B \in \mathcal{P}_B$, and payoff $
    u(\rho_A,\rho_B) = \operatorname{tr}\qty((\rho_A \otimes\rho_B)H)$, the best response is given by, if player $A$ is the minimizing player:
        \begin{equation}
        \begin{aligned}
           \max_{u \in \mathbb{R}, \, \rho_B \in \mathcal{P}_B} \quad & u \\
           \text{s.t.} \quad & \operatorname{tr}_B \left( (I_A \otimes \rho_B) H \right) \succeq u I_A\\
           \min_{u \in \mathbb{R},\, \rho_A \in \mathcal{P}_A} \quad & u \\
           \text{s.t.} \quad & \operatorname{tr}_A \left( (\rho_A \otimes I_B) H \right) \preceq u I_B.
            \label{eq: quantum best response}
        \end{aligned}
    \end{equation}
\end{lemma}

\begin{proof}
    Using the properties of the partial trace
    \begin{equation*}
    \begin{split}
    u(\rho_A,\rho_B) &= \operatorname{tr}\qty((\rho_A \otimes\rho_B)H)\\
    &= \operatorname{tr}_B\qty(\rho_B \operatorname{tr}_A(\rho_A \otimes I_B)H)\\
    &= \operatorname{tr}_A\qty(\rho_A \operatorname{tr}_B(I_A \otimes\rho_B)H)
    \end{split}
    \end{equation*}
    where the operators $\mathcal{K}_B(\rho_A)=\operatorname{tr}_A(\rho_A \otimes I_B)H$ and $\mathcal{K}_A(\rho_B)=\operatorname{tr}_B(I_A \otimes\rho_B)H$ are independent from $\rho_B$ and $\rho_A$, respectively.
    This allows us to write for a fixed $\rho_A$
    \begin{equation*}
        \max_{\rho_B \in \mathcal{P_B}} \operatorname{tr}_B\qty(\rho_B \mathcal{K}_B(\rho_A)) = \lambda_{\max} (\mathcal{K}_B(\rho_A)),
    \end{equation*}
    and for a fixed $\rho_B$
    \begin{equation*}
        \min_{\rho_A \in \mathcal{P}_A} \operatorname{tr}_B\qty(\rho_A \mathcal{K}_A(\rho_B)) = \lambda_{\min} (\mathcal{K}_A(\rho_B)).
    \end{equation*}
    Finally, we can write the minimax and the maximin problems as
    \begin{equation*}
    \begin{split}
        \max_{\rho_B \in \mathcal{P}_B}&  \lambda_{\max} (\operatorname{tr}_B(I_A \otimes \rho_B)H),\\
        \min_{\rho_A \in \mathcal{P}_A}&  \lambda_{\min} (\operatorname{tr}_A(\rho_A \otimes I_B)H)
    \end{split}
    \end{equation*}
    or as the pair of SDPs 
    \begin{equation*}
        \begin{aligned}
           \max_{u \in \mathbb{R}, \, \rho_B \in \mathcal{P}_B} \quad & u \\
           \text{s.t.} \quad & \operatorname{tr}_B \left( (I_A \otimes \rho_B) H \right) \succeq u I_A\\
           \min_{u \in \mathbb{R}, \, \rho_A \in \mathcal{P}_A} \quad & u \\
           \text{s.t.} \quad & \operatorname{tr}_A \left( (\rho_A \otimes I_B) H \right) \preceq u I_B,
        \end{aligned}
    \end{equation*}
    which establishes the required result. \frqed
\end{proof}

Having established the notion of a best response in a game without deception, one naturally asks how a victim should act in a potentially deceptive game. In this context, the concept of a best response becomes less straightforward. While the victim may identify a strategy that guarantees a minimum payoff for the deceptive game specified by $H'$, there is no assurance that this strategy will perform effectively in the actual game $H$. Indeed, the victim's response could span a wide range of behaviors, contingent on their level of awareness and the accuracy of their estimates regarding the deception they face.




We consider two such possible ways for a victim to choose behaviors in a quantum game with deception.

\begin{definition}
A ``naive'' victim assumes that $H'$ is the true payoff Hamiltonian and chooses Nash equilibrium actions for the game it defines. 
That is, selects a 
\begin{equation*}
    (\rho_B)_n \in \Phi(H),
\end{equation*}
where $\Phi(H)$ is defined as the set of density operators that solve~\eqref{eq: quantum best response}.
\hfill\(\square\)
\end{definition}

\begin{definition}
A ``robust'' victim is aware of the possibility of deception, and selects a strategy to optimize the worst-case outcome consistent with the announced payoff Hamiltonian: 
\begin{align*}
    (\rho_B)_r &\in \Phi_r(H') := \Phi(H'-\hat{D}), \\
    \text{s.t.}\, & (\hat{\rho}_A, \hat{D}) \in \argmin_{\rho_A \in \mathcal{P}_A, D \in \mathcal{D}} \operatorname{tr} ((\rho_A \otimes \rho_B) (H' - D)).
\end{align*}
\hfill\(\square\)
\end{definition}

Although each of these models, considered in isolation, may appear to impose a restrictive structure on the victim’s decision-making process, Theorem~\ref{thm:robust_victim} shows that they are behaviorally equivalent. This equivalence establishes that their predictive content is robust to heterogeneity in victim rationality, thereby offering a rigorous justification for Assumption~\ref{asmp:victim_response}.


\begin{theorem}
\label{thm:robust_victim}
Let $H \in \mathbb{C}^{(n_A n_B)\times(n_A n_B)}$ be a payoff Hamiltonian.
Then, one has $\Phi(H)=\Phi_r(H)$.
\end{theorem}
\begin{proof}
    We begin by writing the expected payoff in the quantum game as a function of the quantum states:
    \begin{dmath*}
        u_{\textrm{quantum}} = \mel**{\psi_A \otimes \psi_B} {H} {\psi_A \otimes \psi_B} = \mel**{\psi_A \otimes \psi_B} {(H'-D)} {\psi_A \otimes \psi_B}.
    \end{dmath*}


From the perspective of the victim player, which is the maximizing one, their best response can be defined in the robust response as:
\begin{dmath}
    \Phi_r(H) = \argmax_{\psi_B \in \mathcal{H}_B} \min_{\psi_A \in \mathcal{H}_A, D \in \mathcal{D}}  \mel**{\psi_A \otimes \psi_B} {(H'-D)} {\psi_A \otimes \psi_B} = \argmax_{\psi_B \in \mathcal{H}_B} \left( \min_{\psi_A \in \mathcal{H}_A} \mel**{\psi_A \otimes \psi_B} {H'} {\psi_A \otimes \psi_B} \\ - \max_{D \in \mathcal{D}} \mel**{\psi_A \otimes \psi_B} {D} {\psi_A \otimes \psi_B} \right).
\end{dmath}

For the $\max_{D\in \mathcal{D}} \mel**{\psi_A \otimes \psi_B} {D} {\psi_A \otimes \psi_B}$ term we will show that this is equal to a constant value. $\ket{\psi_A \otimes\psi_B}$ is a unit vector, since it is a Kronecker product of quantum states, so $\norm{{\psi_1}}=\norm{{\psi_2}}=1$ and $\norm{\psi_A \otimes \psi_B}= \norm{{\psi_1}} \norm{{\psi_2}}=1$. 
For unit vectors, and since $D$ is Hermitian, this has the form of the Rayleigh quotient, for which it holds:
\begin{equation*}
    \lambda_{\min}(D)\leq \mel**{\psi_A \otimes \psi_B} {D} {\psi_A \otimes \psi_B} \leq \lambda_{\max}(D)
\end{equation*}
where $\lambda_{\max}$ is the maximum eigenvalue of $D$, or spectral radius of $D$, and $\lambda_{\min}$ is the minimum eigenvalue of $D$. It is also true that $\lambda_{\max}(D) \leq ||D||_p$ for any matrix norm induced by a vector $p$-norm. Therefore, since we assume $||D||_1 \leq \Delta$, we have that $\lambda_{\max}(D) \leq \Delta$. The equality with $\Delta$ can hold for both the upper and lower bounds and for every unit vector. For instance consider $D=\Delta I$, where $\lambda_{\min}=\lambda_{\max}=\Delta$. Therefore, we have the bound $\mel**{\psi_A \otimes \psi_B} {D} {\psi_A \otimes \psi_B} \leq \Delta$ with equality being achieved for appropriate values of $D$, and so we have that:


\begin{equation*}
    \max_{D \in \mathcal{D}}\; \mel**{\psi_A \otimes \psi_B} {D} {\psi_A \otimes \psi_B} = \Delta.
\end{equation*}

Hence we have:

\begin{dmath}
    \Phi_r(H)=\argmax_{\psi_A \in \mathcal{H}_A} \min_{\psi_B \in \mathcal{H}_B} \mel**{\psi_A \otimes \psi_B} {H'} {\psi_A \otimes \psi_B} - \Delta = \argmax \min  \mel**{\psi_A \otimes \psi_B} {H'} {\psi_A \otimes \psi_B} = \Phi(H) 
\end{dmath}

The final equation holds since the inner objectives differ only by a constant. Another way to see that is that this is equivalent with shifting the expected payoffs of a game with a constant, which doesn't affect the strategy derivation for the players.

This completes the proof. \frqed








\end{proof}

This means that a ``robust" victim that is aware of the possibility it is being deceived cannot guarantee better payoffs than a ``naive" victim that is unaware of the deception. This allows for a great simplification of the problem as explained, as now we can derive a program where the leader chooses a deception and a strategy that minimizes their cost assuming that the victim best responds to deceptive quantum game $H'$. An important sidenote is also that this proof works for any induced $p$-norm of $D$, beyond the 1-norm we consider in the context of this paper.

To find the optimal deception in classical games, \cite{gouldNovelFrameworkHoneyX2025} solved the program:
\begin{equation}\label{eq:classical deception}
\begin{split}
    (x_i, D,y_i) \in &\argmin_{x\in\Delta(m), D \in\mathcal{D}, y\in\Delta(n)} v_G(x,y) \\
    &\hspace{1.5cm} \textrm{s.t.} \quad y \in \Phi_B(G+D)
\end{split}
\end{equation}
where $\Phi_B(G)$ denotes a security policy of player $B$ in a game with payoff matrix $G$, where in classical games $v_G(x,y)=x^\top Gy$. We consider a similar problem for quantum games, where the deceiver wants to optimize their payoff given that the victim naively best-responds to the deceptive game $H'$\footnote{Since the victim's best response may not be unique, we make the optimistic assumption that the one most beneficial for the deceiver is selected. This corresponds to the solution concept of a strong Stackelberg equilibrium, and is standard in the literature~\cite{tambeSecurityGameTheory2011, conitzerComputingOptimalStrategy2006}.}. This follows the essence of the deception framework proposed in~\cite{gouldNovelFrameworkHoneyX2025}. In quantum games, considering separability of the quantum states (meaning not fully entangled states) we have that the payoff is given by $\operatorname{tr}\left({(\rho_a \otimes \rho_b)H}\right)$, as explained in Section 2. 
Therefore, the quantum equivalent of \eqref{eq:classical deception} is:
\begin{dmath}
    (\rho_A, D, \rho_B) \in \argmin_{\substack{\rho_A \in \mathcal{P}_A, \\ \rho_B \in \mathcal{P}_B, \\ D\in \mathcal{D}}} \operatorname{tr}\left({(\rho_A \otimes \rho_B)H}\right) \\ \textrm{s.t.} \ {\rho_B \in \Phi_B(H+D)}.
    \label{eq: quantum deception}
\end{dmath}

The problem presented in~\eqref{eq: quantum deception} can be reformulated to the single-level optimization problem~\eqref{main program}, allowing for better computational tractability.

\begin{theorem}

Let $H \in \mathcal{H}_A \otimes\mathcal{H}_b$ be the payoff Hamiltonian of a zero-sum quantum game, $\rho_A \in \mathcal{P}_A \subseteq \mathbb{C}^{n_A \times n_A}, \rho_B \in \mathcal{P}_B\subseteq \mathbb{C}^{n_B \times n_B}$ density matrices for player A and player B of this game respectively, $D\in\mathcal{D}$ the deceptive payoff applied on $H$ by player A and $\Delta \in \mathbb{R_+}$ the deception budget. Then, $\rho_A^\star, D^\star,\rho_B^\star$ that belong in an optimal solution $(\rho_A^\star, D^\star,\rho_B^\star,\Omega^\star,u^\star)$ of the following mathematical program constitute an optimal solution ($\rho_A^\star, D^\star,\rho_B^\star$) of the program in equation \eqref{eq:classical deception}:

\begin{subequations}\label{main program}
\begin{align}
    \min_{\substack{
        \rho_A, \Omega \in \mathbb{C}^{n_A \times n_A}, \\
        \rho_B \in \mathbb{C}^{n_B \times n_B}, \\
        D \in \mathbb{C}^{n_An_B \times n_An_B}, \\
        v \in \mathbb{R}
    }}
    & \operatorname{tr}\left( (\rho_A \otimes \rho_B) H \right) \tag{\theparentequation} \\
    \text{s.t.} \quad 
    & \operatorname{tr}_B \left( (I_A \otimes \rho_B)(H+D) \right) \succeq u I_A \label{eq:1a} \\
    & \operatorname{tr}_A \left( (\Omega \otimes I_B)(H + D) \right) \preceq u I_B \label{eq:1b} \\
    & \|D\|_1 \leq \Delta \label{eq:1c} \\
    & \rho_A, \rho_B, \Omega \succeq 0 \label{eq:1d} \\
    & \operatorname{tr}(\rho_A)=\operatorname{tr}(\rho_B)=\operatorname{tr}(\Omega)=1.\label{eq:1e}
\end{align}
\label{eq: main}
\end{subequations}

\end{theorem}



\begin{proof}
    We begin our proof from equation \eqref{eq: quantum deception}. To solve this we will use the definitions of the security policies in the quantum games. As seen in equation \eqref{eq: quantum best response}, this problem translates to:

\begin{equation*}
\begin{aligned}
    \max_{\rho_A, u} &\quad u\\
    \textrm{s.t.} &\quad \operatorname{tr}_B((I_A \otimes \rho_B)(H+D)) \succeq u I_A, \\
    & \rho_A \succeq 0, \operatorname{tr}(\rho_A)=1. 
\label{second level optimization}
\end{aligned}
\end{equation*}

This, therefore, results in our bi-level optimization program. To avoid this structure, we convert it to bilinear using duality. More specifically, Von Neumann's minimax theorem applies to quantum zero-sum games as well, both in finite \cite{Boukas2000} as well as infinite dimensions \cite{AccardiBoukas2020}. This means that the dual of this problem is the best response of the maximizing player, written as:
\begin{equation*}
\begin{aligned}
   \min_{v, \rho_B, \Omega} \quad & v \\
    \text{s.t.} \quad & \operatorname{tr}_A \left( (\Omega \otimes I_B)(H + D) \right) \preceq vI_B \\
    & \Omega \succeq 0, \operatorname{tr}(\Omega)=1.
\end{aligned}
\end{equation*}

In this formulation, the dual variable $\Omega$ would correspond to the density matrix of player A best responding to the deceptive game $H' = H+D$. Because of the Minimax theorem, the optimal values for the two problems coincide as strong duality holds, meaning $u^\star=v^\star$ and therefore we can combine them both into the single optimization formulation in equation \eqref{main program} to derive the final program. Variable $u$ therefore would represent the calculated expected payoff in equilibrium for the deceptive game $H'$. 

This completes the proof. \frqed


\end{proof}

The solution of the program of equation \eqref{eq: main} would correspond to the optimal deception $D$ of our framework, along with the best response strategy $\rho_A$. It is a bilinear semi-definite program; therefore, it is non-convex, making it quite hard to solve for large games. Furthermore, the solution domain resides in the complex space, which further complicates the problem. However, for small quantum games, this is not necessarily a problem as will be shown in the next section, although the complexity might become prohibitive for larger ones. In the following section, we will present the results of a series of simulations we performed. 
\section{Simulations}\label{sec:simulations}

\begin{figure*}[t]
    \centering
    \begin{subfigure}{0.45\textwidth}
        \centering
        \includegraphics[width=\linewidth]{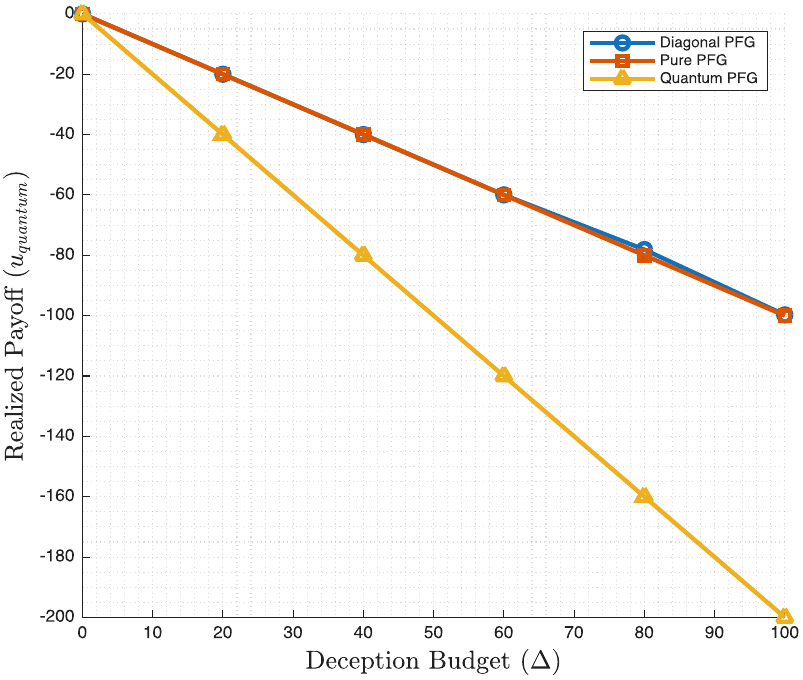}
        \caption{}
        \label{fig:sub-a}
    \end{subfigure}
    \hfill
    \begin{subfigure}{0.45\textwidth}
        \centering
        \includegraphics[width=\linewidth]{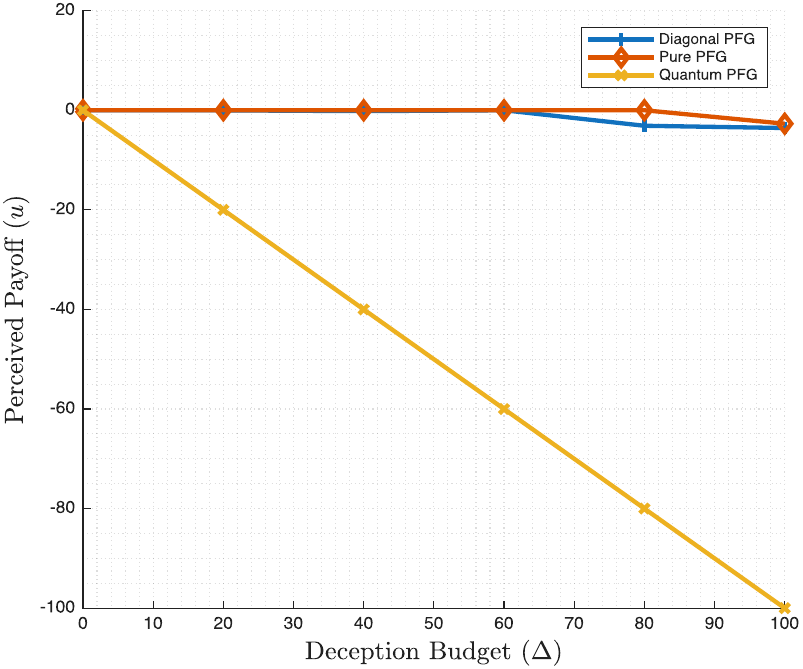}
        \caption{}
        \label{fig:sub-b}
    \end{subfigure}

\caption{Dependence of the deception cost $\Delta$ on \textbf{(a)} the realized payoff of the initial game $H$ from the perspective of the deceiver/minimizer and \textbf{(b)} the perceived payoff $u$ of the deceptive game $H'$ from the perspective of the victim/maximizer, for each of the three games: ``pure", ``diagonal", and ``quantum" PFG. The payoff in \textbf{(a)} decreases (increases in absolute value) approximately linearly with increasing $\Delta$, saturating at the maximum possible absolute value of $100$ for large deception budgets in the first two games, and $200$ for the ``quantum" PFG (the maximum possible absolute value would be $400$ with sufficient deception budget). The optimal perceived payoff $u$ in \textbf{(b)} grows linearly with $\Delta$ for the ``quantum" PFG, but remains close to $0$ for the other games before exhibiting slight fluctuations at larger values.}

 \label{fig:graph}
\end{figure*}

\newcommand{\MatrixCell}[1]{%
  \rule{0pt}{10.0ex}
      $\vcenter{\hbox{%
    \begingroup
      \setlength{\arraycolsep}{2pt}
      \renewcommand{\arraystretch}{1.0}%
      $\begin{bmatrix}#1\end{bmatrix}$%
    \endgroup
  }}$%

  \rule[-8ex]{0pt}{0ex}
}

We implemented the proposed quantum deception framework in \textsc{MATLAB}, utilizing the YALMIP optimization toolbox. As a case study, we consider the quantum version of the Penny Flip Game (PFG) introduced in \cite{wu2005hamiltotianformalismgametheory}. Its classical version is defined by the payoff matrix
\(
A = \mqty[1 & -1 \\ -1 & 1],
\)
with an associated payoff operator
\(
P = \mqty[1 & 0 \\ 0 & -1].
\)
This payoff operator can be incorporated into \eqref{eq:Hqm} to derive either the quantum representation of the game or recover the classical form, as detailed in \cite{wu2005hamiltotianformalismgametheory}. However, the full quantum construction results in a $16 \times 16$ matrix, which is computationally expensive to simulate.  

To make the problem tractable while focusing on validating our framework, we instead consider a smaller subgame. Specifically, we restrict the players’ available actions to the Pauli matrices $\{I,X\}$, which correspond to the classical \textit{“do not flip”} and \textit{“flip”} pure moves. The Hilbert space dimensions are set to $n_A=n_B=2$, so the joint Hilbert space is 
\[
\mathcal{H}_A \otimes \mathcal{H}_B, \quad \dim = N = n_An_B = 4.
\]

Following the procedure described in Section~\ref{subsec:quantum_game}, the resulting payoff Hamiltonian $H \in \mathbb{C}^{N \times N}$ corresponds to a $4 \times 4$ quantum subgame of the PFG \cite{wu2005hamiltotianformalismgametheory}. The reason we select this game is due to its simplicity which allows for explainability in our simulations results, as will be further explored in the corresponding subsection. For numerical stability, the Hamiltonian is scaled by a factor of $100$:
\begin{equation*}
H_{\textrm{pure}} = 
\begin{bmatrix}
 100 & 0    & 0    & 100 \\
   0 & -100 & -100 & 0 \\
   0 & -100 & -100 & 0 \\
 100 & 0    & 0    & 100
\end{bmatrix}.
\end{equation*}

Note that in this restricted quantum subgame spanned by $\{I,X\}$, the Hamiltonian includes non-zero diagonal elements. This property differentiates it from its purely classical counterpart:
\begin{equation*} 
H_{\textrm{diagonal}} =
\begin{bmatrix}
 100 & 0    & 0    & 0 \\
   0 & -100 & 0    & 0 \\
   0 & 0    & -100 & 0 \\
   0 & 0    & 0    & 100
\end{bmatrix}.
\end{equation*}
Note that this $H_{\textrm{diagonal}}$ lacks the off-diagonal elements that would interact with the quantum strategies $\mathcal{S}$, rendering the game equivalent to a classical one expressed in our quantum framework.

\begin{table*}[!t]
\centering
\caption{$D$ matrices with respect to $\Delta$ (rounded to 1 decimal place) for each of the 3 quantum payoff Hamiltonians.}
\renewcommand{\arraystretch}{1.15}
\setlength{\tabcolsep}{2pt}
\tiny
\resizebox{\textwidth}{!}{%
\begin{tabular}{c*{3}{c}}\toprule
 $\Delta$ & $20.0$ & $40.0$ & $60.0$ \\ \midrule
Pure & \MatrixCell{0.0+0.0j & -0.0-0.0j & -0.0+1.6j & -3.7-18.0j \\ -0.0+0.0j & -0.0+0.0j & 3.6-18.0j & -0.1+1.6j \\ -0.0-1.6j & 3.6+18.0j & 0.0+0.0j & 0.0-0.0j \\ -3.7+18.0j & -0.1-1.6j & 0.0+0.0j & -0.0+0.0j} & \MatrixCell{0.0+0.0j & -0.0+0.0j & -0.3-22.3j & -7.1+16.3j \\ -0.0-0.0j & 0.0+0.0j & 7.1+16.3j & -0.1-22.3j \\ -0.3+22.3j & 7.1-16.3j & 0.0+0.0j & -0.0+0.0j \\ -7.1-16.3j & -0.1+22.3j & -0.0-0.0j & -0.0+0.0j} & \MatrixCell{0.0+0.0j & 0.0+0.0j & 1.1+35.6j & -14.6+19.6j \\ 0.0-0.0j & 0.0+0.0j & 14.6+19.6j & -1.1+35.6j \\ 1.1-35.6j & 14.6-19.6j & -0.0+0.0j & -0.0-0.0j \\ -14.6-19.6j & -1.1-35.6j & -0.0+0.0j & -0.0+0.0j} \\ \hline
Diagonal & \MatrixCell{19.5+0.0j & 0.0-0.5j & -0.0-0.0j & -0.0+0.0j \\ 0.0+0.5j & 19.5+0.0j & 0.0-0.0j & -0.0-0.0j \\ -0.0+0.0j & 0.0+0.0j & -20.0+0.0j & 0.0+0.0j \\ -0.0-0.0j & -0.0+0.0j & 0.0-0.0j & -20.0+0.0j} & \MatrixCell{39.6+0.0j & -0.4+0.1j & 0.0-0.0j & -0.0+0.0j \\ -0.4-0.1j & 39.6+0.0j & -0.0+0.0j & 0.0-0.0j \\ 0.0+0.0j & -0.0-0.0j & -39.8+0.0j & 0.2-0.0j \\ -0.0-0.0j & 0.0+0.0j & 0.2+0.0j & -39.8+0.0j} & \MatrixCell{-60.0+0.0j & 0.0-0.0j & -0.0-0.0j & 0.0-0.0j \\ 0.0+0.0j & -60.0+0.0j & -0.0+0.0j & -0.0-0.0j \\ -0.0+0.0j & -0.0-0.0j & 60.0+0.0j & -0.0+0.0j \\ 0.0+0.0j & -0.0+0.0j & -0.0-0.0j & 60.0+0.0j} \\ \hline
Quantum & \MatrixCell{0.0+0.0j & -0.0+0.0j & -0.0+20.0j & 0.0+0.0j \\ -0.0-0.0j & -0.0+0.0j & -0.0+0.0j & 0.0+20.0j \\ -0.0-20.0j & -0.0-0.0j & 0.0+0.0j & 0.0+0.0j \\ 0.0-0.0j & 0.0-20.0j & 0.0-0.0j & -0.0+0.0j} & \MatrixCell{0.0+0.0j & 0.0+0.0j & 0.0+40.0j & 0.0+0.0j \\ 0.0-0.0j & -0.0+0.0j & -0.0+0.0j & -0.0+40.0j \\ 0.0-40.0j & -0.0-0.0j & 0.0+0.0j & 0.0+0.0j \\ 0.0-0.0j & -0.0-40.0j & 0.0-0.0j & -0.0+0.0j} & \MatrixCell{0.0+0.0j & 0.0+0.0j & -0.0+60.0j & 0.0+0.0j \\ 0.0-0.0j & -0.0+0.0j & -0.0+0.0j & 0.0+60.0j \\ -0.0-60.0j & -0.0-0.0j & 0.0+0.0j & 0.0+0.0j \\ 0.0-0.0j & 0.0-60.0j & 0.0-0.0j & -0.0+0.0j} \\ \bottomrule
\end{tabular}%
}

\vspace{0.8em}

\tiny
\begin{tabular}{c*{3}{c}}
\toprule
 $\Delta$ & $80.0$ & $100.0$ \\ \midrule
Pure & \MatrixCell{-0.0+0.0j & 0.1+0.1j & 4.8-55.0j & -17.7-17.2j \\ 0.1-0.1j & 0.0+0.0j & 12.3-19.6j & 5.9-55.1j \\ 4.8+55.0j & 12.3+19.6j & -1.5+0.0j & 0.0+0.0j \\ -17.7+17.2j & 5.9+55.1j & 0.0-0.0j & 0.0+0.0j} & \MatrixCell{-0.0+0.0j & 0.0+0.0j & -6.6+53.8j & -45.7+1.5j \\ 0.0-0.0j & -0.0+0.0j & 45.7+1.5j & -0.5+54.2j \\ -6.6-53.8j & 45.7-1.5j & 0.0+0.0j & 0.0+0.0j \\ -45.7-1.5j & -0.5-54.2j & 0.0-0.0j & 0.0+0.0j} \\ \hline
Diagonal & \MatrixCell{80.0+0.0j & 0.0+0.0j & -0.0+0.0j & -0.0-0.0j \\ 0.0-0.0j & 80.0+0.0j & 0.0+0.0j & -0.0+0.0j \\ -0.0-0.0j & 0.0-0.0j & -80.0+0.0j & 0.0-0.0j \\ -0.0+0.0j & -0.0-0.0j & 0.0+0.0j & -80.0+0.0j} & \MatrixCell{-45.8+0.0j & 0.0+0.0j & 40.8+16.4j & 0.0+0.0j \\ 0.0-0.0j & 98.1+0.0j & -0.0-0.0j & 1.7+0.8j \\ 40.8-16.4j & -0.0+0.0j & -44.5+0.0j & -0.0+0.0j \\ 0.0-0.0j & 1.7-0.8j & -0.0-0.0j & -98.1+0.0j} \\ \hline
Quantum & \MatrixCell{0.0+0.0j & -0.0+0.0j & -0.0+80.0j & 0.0-0.0j \\ -0.0-0.0j & -0.0+0.0j & -0.0-0.0j & 0.0+80.0j \\ -0.0-80.0j & -0.0+0.0j & 0.0+0.0j & -0.0+0.0j \\ 0.0+0.0j & 0.0-80.0j & -0.0-0.0j & -0.0+0.0j} & \MatrixCell{0.0+0.0j & 0.0+0.0j & 0.0+100.0j & 0.0+0.0j \\ 0.0-0.0j & -0.0+0.0j & -0.0+0.0j & -0.0+100.0j \\ 0.0-100.0j & -0.0-0.0j & 0.0+0.0j & 0.0+0.0j \\ 0.0-0.0j & -0.0-100.0j & 0.0-0.0j & -0.0+0.0j} \\ \bottomrule
\end{tabular}%
\label{tab:D_matrices_vertical_4x4}
\end{table*}

The final game used in our simulations is the restricted quantum subgame spanned by $\{I,Z\}$, so the elements of the Pauli set used to generate this game can be thought of as the \textit{do not flip} $\{I\}$ and a quantum strategy with no classical analog $\{Z\}$. The latter has an effect of manipulating the phase of the quantum object used to describe the moves of each player.  In this formulation, the maximum achievable payoff is $400$, compared to $100$ in the earlier cases:
\begin{equation*} 
H_{\textrm{quantum}} =
\begin{bmatrix}
 100 & -100j   & -100j   & 100 \\
 100j & -100   & -100    & -100j \\
 100j & -100   & -100    & -100j \\
 100  & 100j   & 100j    & 100
\end{bmatrix}.
\end{equation*}

The optimization variables are the players’ density matrices $\rho_A, \in \mathcal{P}_A, \rho_B \in \mathcal{P}_B$, the dual density matrix $\Omega \in \mathcal{P}_A$, the equilibrium scalar value $v \in \mathbb{R}$, and the deception matrix $D \in \mathbb{H}^{4}$. The deception matrix is constrained by an $\ell_1$-norm budget $\|D\|_1 \leq \Delta$, as described in previous sections. All density matrices are positive semidefinite and normalized to have unit trace.  

The optimization objective is to minimize the deceiver’s expected payoff:
\[
\min \; \operatorname{tr}\!\big((\rho_A \otimes \rho_B)H\big),
\]
using the program defined in \eqref{main program}.

We considered deception budgets $\Delta \in \{0,20,\dots,100\}$. For each value of $\Delta$, the semidefinite program was solved using the global branch-and-bound solver \texttt{bmibnb}, with \texttt{sdpt3} as the lower-bounding SDP solver and \texttt{fmincon} as the upper-bounding local solver. Solver tolerances were fixed at $10^{-3}$ (absolute gap) and $10^{-2}$ (relative gap). The maximum number of iterations was set to $100$, with a runtime cap of $30$ minutes per instance.


The performance was evaluated by plotting both the deceiver’s payoff and the victim’s perceived payoff $u$ as functions of the deception budget $\Delta$, as shown in Figure~\ref{fig:graph}. Across all three games, the expected payoff for the minimizing player~A decreases approximately linearly (increasing in absolute value) as $\Delta$ grows. This trend is consistent with theoretical expectations and suggests that, even though the algorithm may not have reached full convergence, the computed solutions are either optimal but not yet verified by the upper solver, or at least very close to optimal.

The corresponding deception matrices $D$ for selected values of $\Delta$ are reported in Table~\ref{tab:D_matrices_vertical_4x4}, which further supports this interpretation. For example, in the ``diagonal'' PFG, the optimal deception expends the entire budget to shift payoff mass away from the outcome favorable to player~B and toward outcomes favorable to player~A. This structure is evident in the diagonal entries of $D$, with any small off-diagonal components being of lower order of magnitude. These residual entries have therfore negligible effect on payoffs and are indicative of near-optimality. In this setting, player~B is deceived into believing their maximum attainable payoff is~$0$ (Figure \ref{fig:sub-b}), while player~A’s actual payoff improves (Figure \ref{fig:sub-a}). Similar patterns can be observed in the other two games. The ``quantum'' PFG exhibits distinctive deception involving selected off-diagonal entries, while the ``pure'' PFG relies on a more intricate deception pattern, primarily along the anti-diagonal and certain off-diagonal elements. 

Two further observations are noteworthy. First, when $\Delta = 100$, the ``diagonal'' PFG produces a deception matrix $D$ that is no longer strictly diagonal but includes non-trivial off-diagonal components. Second, the ``pure'' PFG also yields deception matrices with imaginary off-diagonal structure, even though its original payoff Hamiltonian contains only real off-diagonal elements. 

These results highlight a fundamental distinction between quantum and classical deception. In quantum settings, the emergence of off-diagonal entries in the payoff operator signals convergence toward an inherently quantum solution. Such entries encode both magnitude and phase, introducing interference effects that reshape outcome probabilities \cite{Accardi20212}. This mechanism enables constructive and destructive interference, allowing superpositional strategies to exploit phase relationships in ways that classical diagonal payoffs cannot. Classical games, restricted to real diagonal entries, only permit amplitude adjustments and thus lack this additional degree of strategic control. Consequently, the complex off-diagonal structure provides a uniquely quantum mechanism for manipulating the game’s probability landscape. This distinction reflects the non-Bayesian nature of quantum probability, which admits phenomena beyond the expressive power of classical Bayesian frameworks \cite{accardi1981topics}.

\section{Conclusion}\label{sec:conclusion}

In this paper, we developed a novel extension of the Honey-X deception framework into the quantum game-theoretic domain. To the best of our knowledge, this is the first work to study deceptive quantum games of any kind. We introduced a model in which a deceiver perturbs the payoff Hamiltonian within a bounded budget, and the victim best responds to the deceptive game. We proved that, as in the classical setting, the equilibrium strategy of a naive victim coincides with that of a robust victim who anticipates possible deception, thereby simplifying the analysis of victim behavior. Exploiting this equivalence, we formulated the quantum deception problem as a bilinear semidefinite program, enabling the computation of optimal deceptive strategies. Our simulations on quantum extensions of the Penny Flip game illustrated how the enlarged quantum strategy space supports richer deceptive behaviors and can yield greater adversarial advantage compared to classical settings. Taken together, these results establish both a theoretical and computational foundation for the study of deception in quantum strategic interactions.

Several research directions follow naturally from our work. One avenue is to extend the framework to incorporate fully entangled states, which would allow for more expressive strategy spaces and potentially stronger forms of deception. Another is to explore alternative models of deception beyond bounded perturbations of the Hamiltonian. A particularly important direction is improving computational efficiency: while our reformulation allows for more off-the-self solvers to be used in it rather than the initial bilevel program, the classical-to-quantum lift still introduces a complexity of $O(n^4)$. Also, the solution domain involves inherently imaginary numbers that further amplify computational cost. Significant advances will therefore be needed to make the framework scalable for large practical games, possibly leveraging quantum algorithms themselves to accelerate optimization.

\balance
\bibliographystyle{IEEEtran}
\bibliography{references,refs,gould_ref}


\end{document}